 \documentclass[12pt,draftclsnofoot,journal,onecolumn]{IEEEtran}

\usepackage[dvips]{color}
\usepackage{epsf}
\usepackage{times}
\usepackage{epsfig}
\usepackage{graphicx}
\usepackage{amsmath}
\usepackage{amssymb}
\usepackage{amsxtra}
\usepackage{here}
\usepackage{rawfonts}
\usepackage{times}
\usepackage{url}
\usepackage{cite}

\newtheorem{theorem}{\bf Theorem}

\newtheorem{proposition}{\bf Proposition}

\newtheorem{definition}{\bf Definition}

\IEEEoverridecommandlockouts
\begin{document}

\title{{\huge A Cooperative Bayesian Nonparametric Framework for Primary User Activity Monitoring in Cognitive Radio Networks} \vspace{-0.3cm}}
\author{
{ Walid Saad$^1$, Zhu Han$^2$,  H. Vincent Poor$^1$, Tamer Ba\c{s}ar$^3$, and Ju Bin Song$^4$
}\\
\small
$^1$ Electrical Engineering Department, Princeton University, Princeton, NJ, USA, \url{{saad, poor}@princeton.edu}. \\
$^2$ Electrical and Computer Engineering Department, University of Houston, Houston, TX, USA, \url{zhan2@mail.uh.edu}.\\
$^3$ Coordinated Science Laboratory, University of Illinois at Urbana-Champaign, USA, Email: \url{basar1@illinois.edu}.\\
$^4$ College of Electronics and Information, Kyung Hee University, South Korea, Email: \url{jsong@khu.ac.kr}.

\thanks{This work was in part supported by the Air Force Office of Scientific Research Grant FA9550-08-1-0480, in part by the National Science Foundation under Grants CNS-0910461, CNS-0953377, CNS-0905556, ECCS-1028782; and in part by the National Research Foundation of Korea (Grant No.20090075107).}}
\maketitle
\vspace{-1.6cm}
\thispagestyle{empty}
\begin{abstract}
This paper introduces a novel approach that enables a number of cognitive radio devices that are observing the availability pattern of a number of primary users~(PUs), to cooperate and use \emph{Bayesian nonparametric} techniques to estimate the distributions of the PUs' activity pattern, assumed to be completely unknown. In the proposed model, each cognitive node may have its own individual view on each PU's distribution, and, hence, seeks to find partners having a correlated perception. To address this problem,  a coalitional game is formulated between the cognitive devices and an algorithm for cooperative coalition formation is proposed. It is shown that the proposed coalition formation algorithm allows the cognitive nodes that are experiencing a similar behavior from some PUs to self-organize into disjoint, independent coalitions. Inside each coalition, the cooperative cognitive nodes use a combination of Bayesian nonparametric models such as the Dirichlet process and statistical goodness of fit techniques in order to improve the accuracy of the estimated PUs' activity distributions. Simulation results show that the proposed algorithm significantly improves the estimates of the PUs' distributions and yields a performance advantage, in terms of reduction of the average achieved Kullback-–Leibler distance  between the real and the estimated distributions, reaching up to $36.5\%$ relative the non-cooperative estimates. The results also show that the proposed algorithm enables the cognitive nodes to adapt their cooperative decisions when the actual PUs' distributions change due to, for example, PU mobility.

\end{abstract}

\newpage
 \setcounter{page}{1}
\section{Introduction}
Cognitive radio has been proposed as a novel communication paradigm that allows an efficient sharing of the under-utilized radio spectrum resources between licensed or primary users~(PUs) and unlicensed or secondary users~(SUs)~\cite{CR01,CR02}. Cognitive radio networks are based upon flexible spectrum management techniques that allow licensed and unlicensed users to share the spectrum, while avoiding collisions with one another. The main enablers of such cognitive communications are smart SU devices that can intelligently and dynamically monitor the spectrum, operating only when the PUs are inactive and making sure to vacate the spectrum whenever a PU starts its transmission. Hence, one of the key challenges faced in the design of cognitive radio networks is to ensure dynamic spectrum sharing while maintaining a conflict-free coexistence between primary and secondary users~\cite{CR02,XS04,XS03,XS00,XS01,SA00,Chaud,Lund,Axel}.

In order to detect the PUs' activity, the SUs are typically equipped with sensing capabilities (e.g., energy or cyclostationarity detectors) that enable them to autonomously detect unoccupied spectrum and transmit opportunistically - e.g.,~\cite{Chaud,Lund,SA00,DT00,Axel,XS00,XS01} (see \cite{Axel} for a comprehensive review). Spectrum sensing is a key step for deploying robust cognitive radio networks and has received significant attention~\cite{CS00,CS03,CS07,CS08,CS02,DT00,XS00,XS01,SA00,CS04,CS05,CS06,Chaud,Lund,Axel}. In particular, advanced spectrum sensing techniques such as \emph{cooperative sensing} have been proposed in~\cite{CS00,CS03,CS07,CS08,CS02,CS04,CS05,CS06,Lund} so as to improve the SU's detection capability. The main idea of cooperative sensing is to combine different SU observations so as to have a better decision on whether a PU is present or not, at a given time instant. In \cite{CS00}, the authors propose centralized schemes enabling the SUs to share their sensing decisions given a known PU distribution. The work in \cite{CS03} studies the impact of reporting channel errors on collaborative sensing. In \cite{CS07} and \cite{CS08}, the use of relaying techniques for improving cooperative sensing is thoroughly analyzed. Other performance aspects of cooperative sensing are studied in \cite{Lund,CS04,CS05,CS06,CS02}.

 However, performing cooperative or non-cooperative sensing is known to be a time consuming process that can affect the access performance of the SUs, notably in multi-channel networks~\cite{SA00,CR02,XS04}.  To overcome this problem, recent research activities brought forward the idea of providing, using \emph{control channels}, spectrum monitoring assistance to the SUs so as to improve their performance~\cite{CH01,CH03,CH04,CH06,XS02}. These channels can be used in conjunction with advanced techniques such as cooperative spectrum sensing so as to provide additional information to the SUs that can improve their sensing decisions. For example, the authors in \cite{CH01} studied how a Common Spectrum Coordination Channel~(CSCC) can be used to announce radio and service parameters to the SUs. More recently, the Cognitive Pilot Channel~(CPC) has been introduced \cite{CH03,CH04,CH06,XS02} as a control channel that can convey critical information to the SUs, allowing them to enhance their sensing and access decisions, notably in the presence of multiple channels (i.e., PUs) and access technologies.

Essentially, the CPC is a channel that can carry different information such as estimates of the activity of the PUs, frequency or geographical data, that the SUs can use to improve their sensing, to avoid scanning the entire spectrum for finding spectral holes, and to get a better perception of their environment (e.g., locations and frequencies of the PUs)~\cite{CH03} and \cite{CH04}. Deploying the CPC in a practical network can be done either using existing infrastructure (e.g., existing cognitive users or base stations) or by installing dedicated nodes that carry CPC data, i.e., \emph{CPC nodes}. For transmitting the CPC data, a variety of methods can be used, as proposed in \cite{CH03,CH04,CH06}. 


The use of CPCs and cooperative spectrum sensing have received considerable attention in the research community. However, on the one hand, most of the existing work on CPC deployment such as~\cite{CH03,CH04,CH06,XS04} has focused on implementation and transmission aspects. On the other hand, existing cooperative spectrum sensing techniques such as in \cite{CS00,CS02,CS03,CS04,CS05,CS06,CS07} often assume that the PU's activity follows a certain known or assumed distribution. However, no work seems to have investigated how cognitive device such as CPC nodes can be used to provide information on the activity of the PUs in a practical cognitive network. This primary user activity information can be used, subsequently, to improve the decisions of both cooperative and non-cooperative spectrum sensing. To operate efficiently, the SUs must obtain a good overview of the activity of the PUs, so as to access the spectrum at the right time and for a suitable duration. Moreover, this information is important to improve the cooperative decisions for collaborative sensing techniques such as in ~\cite{CS00,CS02,CS03,CS04,CS05,CS06,CS07}. The objective of this paper is to leverage the use of control channels such as the CPC in order to convey to the SUs accurate estimates of the distribution of the activity of the PUs, which is often sporadic and unknown. In addition, a given PU channel can be seen differently by CPC nodes positioned in different locations of a cognitive network. In most cooperative sensing or CPC literature, these different PU views are often simplified or assumed to be fixed. However, in practice, this assumption may not hold due to a variety of factors such as the locations of the PU transmitters or their power capabilities.  Therefore, developing efficient schemes that allow the cognitive nodes to obtain (e.g., through a CPC) accurate estimates of the PUs' channel availability patterns is a challenging task that is of central importance in maintaining a conflict-free environment between SUs and PUs. To the best of our knowledge, this paper is the first that treats this problem, notably from a cooperative approach that uses Bayesian nonparametric inference as well as game theoretic techniques.

The main contribution of this paper is to introduce a novel cooperative approach between cognitive devices such as CPC nodes that allows them to share their observations on the distributions of the PUs' activity, and, subsequently, build an accurate estimate of these distributions. In particular, given a number of PUs whose availability is perceived differently by a number of CPC nodes, we propose a scheme that allows these nodes to cooperate in order to  estimate the distributions of the PUs' activity, assumed to be completely unknown. In this context, we formulate a coalitional game between the CPC nodes and we develop a suitable coalition formation algorithm. The proposed game allows the CPC nodes to decide, in a distributed manner, on whether to cooperate or not, based on a utility that captures the gain from cooperation, in terms of an improved estimate of the PUs' distributions, and a cost for coordination. Each group of CPC nodes that decides to cooperate and form a coalition will subsequently use \emph{Bayesian nonparametric techniques}, based on the Dirichlet process, as well as goodness of fit statistical tests, to cooperatively infer the perceived distributions of the PUs' activity. We show that, by performing coalition formation, the CPC nodes self-organize into a network of disjoint and independent coalitions that form a Nash-stable partition in which each node has a significantly improved estimate of all the PUs' activity. Simulation results show that the proposed cooperative approach yields a significant performance improvement.

The remainder of this paper is organized as follows: Section~\ref{sec:sysmodel} presents the system model. In Section~\ref{sec:hedonic}, we present the proposed cooperative Bayesian nonparametric scheme and we model it using coalitional game theory. In Section~\ref{sec:algo}, we propose an algorithm for distributed coalition formation. Simulation results are analyzed in Section \ref{sec:sim} and conclusions are drawn in Section \ref{sec:conc}.\vspace{-0.2cm}

\section{System Model}\label{sec:sysmodel}
Consider a network of $N$ cognitive radio devices that are seeking to transmit, opportunistically, over $K$~channels that represent a number of PUs. These devices can be either SUs, fixed secondary base stations, or other fixed or mobile cognitive radio nodes. One typical example of these devices would be a number of cognitive nodes dedicated to provide information to the SUs, e.g., CPC-carrying nodes. Hereinafter, for brevity, we use the term CPC or CPC node to refer to any such cognitive node. The set of all CPCs is denoted by $\mathcal{N}$ while the set of PUs is denoted by $\mathcal{K}$. At any point in time, from the perspective of any CPC node $i\in \mathcal{N}$ (and the SUs in its vicinity), every PU $k \in \mathcal{K}$ is considered to be active, i.e., its channel is occupied, with a probability $\theta_{ik}$. For a given PU $k \in \mathcal{K}$, two distinct CPC nodes $i,j \in \mathcal{N},\ i\neq j$ can see a different value of the probability that $k$ is active, i.e., $\theta_{ik}\neq \theta_{jk}$, depending on various factors such as the distance to the PU, wireless channel fading, or the PU's transmission capabilities. For example, from the perspective of a CPC $i$ that is in the vicinity of a PU $k$, even when the PU uses a small power for transmission, PU $k$'s channel is still seen as being occupied due to the small path loss (or fading) between CPC $k$ and PU $i$. In contrast, from the point of view of another CPC $j$ that is located far away from the same PU $k$, the channel used by PU $k$ appears to be vacant whenever PU $k$'s transmit power is low. The main reason behind these different observations is that, unlike CPC node $i$, CPC node $j$ and the SUs that it serves experience a low interference from a PU $k$ located at a relatively large distance and whose transmit power is attenuated by a reasonably significant channel fading. As a result, from the perspective of CPC node $j$, PU $k$ channel's would be seen as vacant even when it appears occupied to CPC node $i$. In such an illustrative scenario, for the same PU $k$, we would have $\theta_{ik} > \theta_{jk}$. 

Often, the PUs can change their pattern of activity depending on many random parameters, e.g., due to their nature or capabilities. For example, when the PUs represent the mobile nodes of a wireless system (e.g., an LTE or 3G system), they may frequently change their activity depending on the time of the day or the region in which they operate. Hence, for a given PU $k$, the value of the probability $\theta_{ik}$ from the perspective of any CPC $i\in \mathcal{N}$, can be seen as a random variable having a certain probability distribution $P_{ik}(\theta_{ik})$ which is a probability density function over the state space ${\Theta}=[0,1]$ of $\theta_{ik},\ \forall i \in \mathcal{N},\ k \in \mathcal{K}$. Moreover, we consider that the CPCs in $\mathcal{N}$ have \emph{no} prior knowledge on the distribution of the PUs' activity. Thus, for \emph{any} CPC $i\in \mathcal{N}$ and any PU $k \in \mathcal{K}$, the actual real distribution $P_{ik}(\theta_{ik})$ is \emph{completely unknown} by the CPC. Hereinafter, for brevity, we use the term the expression \emph{distribution of the PUs} or \emph{PUs distribution} to refer to the distribution of the PUs' activity/availability.

Each CPC $i\in \mathcal{N}$ performs a limited number of $L_{ik}$ observations $\mathcal{L}_{ik}=\{\theta_{ik}^{1},\ldots,\theta_{ik}^{L_{ik}}\}$  for every PU channel $k \in \mathcal{K}$ so as to get an estimate of the distributions $P_{ik}(\theta_{ik})$. Each observation $\theta_{ik}^t\in\mathcal{L}_{ik}$ is a value for the probability $\theta_{ik}$ observed at a time period $t$.  To obtain $\mathcal{L}_{ik}$ for a channel $k$, a CPC needs to monitor, over a given period of time $t$, the activity of PU $k$ and record the resulting probability $\theta_{ik}^{t}$. This process can be seen as a sampling of the PU's activity distribution. We note that the time period $t$ during which a single observation is recorded must be reasonably large so as to enable the cognitive device to record a reasonably accurate observation. In practice, the exact value for this period is dependent on the cognitive network's implementation, the nature of the PU (whether it is a television transmitter, a mobile device, a base station, etc.) and can be adjusted by the CPC  accordingly. Due to this, the number of observations $L_{ik}$ for each PU channel $k$ is, in practice, small, due to the time consuming nature of this process. The small value of $L_{ik}$ is further corroborated by the fact that, in addition to PU activity estimation, a CPC also needs to perform other tasks such as acquiring frequency and geographical information, and, hence, it cannot dedicate all of its resources to the PU activity estimation process. We note that,  in a given time period, the observations $\mathcal{L}_{ik}$ are the only information that a CPC node $i$ has about the behavior of PU $k$.

Having recorded the observations $\mathcal{L}_{ik}$, each CPC $i \in \mathcal{N}$ must infer the distribution of every PU $k \in \mathcal{K}$.  Given $\mathcal{L}_{ik}$, a CPC $i$ can predict the distribution of the next observation $\theta^{L_{ik}+1}_{ik}$ using the following expression:
\begin{equation}\label{eq:discnc}
H_{ik}(\theta^{L_{ik}+1}_{ik}\in \mathcal{A}|\theta_{ik}^{1},\ldots,\theta_{ik}^{L_{ik}})= \frac{\sum_{l=1}^{L_{ik}}\delta_{\theta_{ik}^l}(\mathcal{A})}{L_{ik}},
\end{equation}
where $\mathcal{A} \subseteq {\Theta}$ is a subset of the space ${\Theta}$ and $\delta_{\theta_{ik}^l}$ is the point mass located at $\theta_{ik}^l$ such that $\delta_{\theta_{ik}^t}(\mathcal{A})=1$ if $\theta_{ik}^t \in \mathcal{A}$ and $0$ otherwise.

When acting non-cooperatively, each CPC can compute the distribution of  $\theta^{L_{ik}+1}_{ik}$ using (\ref{eq:discnc}) which is discrete. Given the limited number of observations $L_{ik}$, using (\ref{eq:discnc}) can yield a large inaccuracy in the estimation. In order to get a more accurate, continuous estimate $\tilde{H}_{ik}$  of the distribution ${H}_{ik}$ in (\ref{eq:discnc}), each CPC $i$ can adopt \emph{kernel density estimation} or kernel smoothing techniques \cite{KDE00,KDE01,KDE02}. As explained in \cite{KDE00}, kernel density estimation methods are popular nonparametric estimators used to draw inferences about a certain distribution based on finite data samples. Kernel density estimation methods aim at smoothing a discrete function in four main steps \cite{KDE01}: (i)- Choosing a \emph{kernel} function which is a symmetric but not necessarily positive continuous function that integrates to one and a scaling factor commonly known as \emph{bandwidth} which controls the smoothness of the estimate, (ii)- Placing the center of the chosen kernel over each observed data point, (iii)- Spreading the influence of each data point over its neighborhood, and, (iv)- Summing the contributions from each data point in order to generate the final estimate.

 One popular kernel density estimator is the Gaussian kernel estimator in which the kernel is chosen as a Gaussian distribution whose bandwidth is its mean~\cite{KDE01}. Then, this kernel is convoluted to the discrete function (or the observations) so as to generate the density estimate. While a detailed treatment of kernel density estimation techniques is beyond the scope of this paper,\footnote{The interested reader is referred to \cite{KDE00} or \cite{KDE01} for further information.} for the proposed model, we assume that, when acting non-cooperatively, the CPCs utilize the generic kernel density estimation via the linear diffusion approach of \cite{KDE02}, in order to obtain a continuous version $\tilde{H}(\theta^{L_{ik}+1}_{ik}\in \mathcal{A}|\theta_{ik}^{1},\ldots,\theta_{ik}^{L_{ik}}))$ of (\ref{eq:discnc}) which constitutes the non-cooperative \emph{kernel estimate}. Note that other kernel estimation techniques can also be adopted, without loss of generality.

Non-cooperatively, the kernel estimate is the most reasonable estimate that any CPC $i\in \mathcal{N}$ can obtain, given its limited number of observations. However, as the number $L_{ik}$ of available observations is generally small, the kernel estimate of the PUs distributions may not perform as well as required by the cognitive network. Therefore, the CPCs need to seek alternative methods to improve their estimate of the PUs distributions without a need for continuous and real-time observation of the PUs' behavior. One approach to solve this problem, which we introduce in this paper, is to let the CPCs interact and cooperate, when possible, in order to improve their perception of the PUs distributions. In particular, CPCs that are observing similar PUs distributions would have an incentive to form cooperative groups, i.e., coalitions, so as to share observations and improve their estimates.

\begin{figure}[!t]
\begin{center}
\includegraphics[width=10cm]{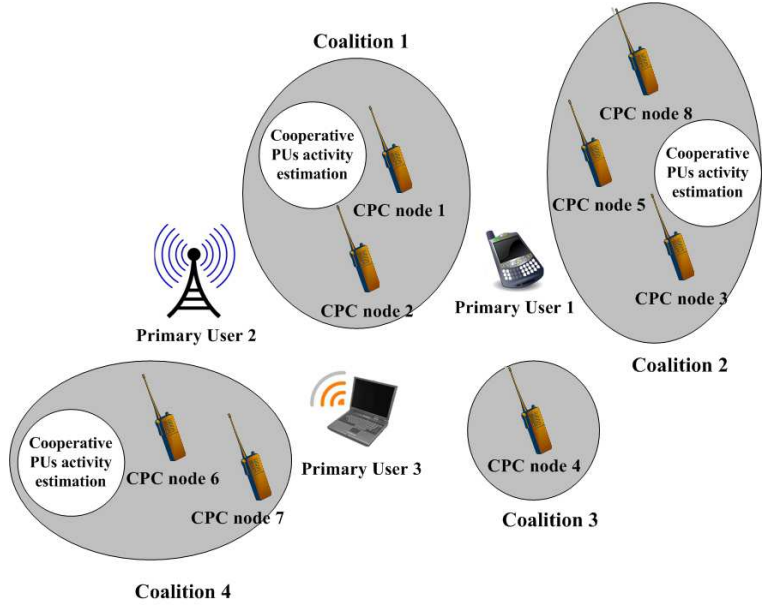}
\end{center}\vspace{-0.6cm}
\caption {An illustration showing a cognitive radio network with $8$ cooperative CPC nodes that form coalitions and can jointly estimate the distributions of $3$~primary users.} \label{fig:ill}\vspace{-0.15cm}
\end{figure}
In Figure~\ref{fig:ill}, we illustrate how cognitive nodes that have a somewhat similar view on the existing PUs group together into coalitions, so as to perform cooperative estimation of the PUs' distributions, for a network with $8$~CPCs and $3$~different PUs. For example, in Figure~\ref{fig:ill}, as CPCs $6$ and $7$ are located almost symmetrically around the $3$~PUs, they find it beneficial to cooperate and share their observations. Similarly, the other CPCs in Figure~\ref{fig:ill} make their cooperative decisions, depending on the correlation between their perceived PUs distributions as well as the potential of having a better estimate. To perform cooperative estimation such as in Figure~\ref{fig:ill}, several challenges must be overcome such as correctly determining whether the cooperative CPCs are observing similar distributions as well as identifying the benefits and costs from cooperation. In this respect, the next section will propose novel solutions to the aforementioned challenges.

\section{Cooperative Bayesian Nonparametric Estimation of Primary Users' Activity}\label{sec:hedonic}
To model the problem illustrated in Figure~\ref{fig:ill}, we will use the analytical tools of coalitional game theory \cite{Game_theory2,WS00}. In particular, we are interested in formulating the proposed CPC cooperation problem as a coalitional game with a non-transferable utility defined as follows \cite[Chap. 9]{Game_theory2}:
\begin{definition}
A coalitional game with \emph{non-transferable utility} consists of a pair $(\mathcal{N},V)$ in which $\mathcal{N}$ represents the set of players and $V$ is a mapping that assigns for any coalition $S \subseteq \mathcal{N}$ a set of payoff vectors that the members of $S$ can achieve. $V(S)$ is a closed and convex subset of $\mathbb{R}^{S}$ .
\end{definition}

In the proposed model, one can see that the set of players $\mathcal{N}$ is the set of CPC nodes. In order to completely describe the coalitional game between the CPCs, our next step is to determine the mapping $V$ which reduces to proposing a utility function that captures the gains and costs that each CPC achieves when cooperating within a certain coalition $S \subseteq \mathcal{N}$. To determine this utility function, we must first provide a cooperative procedure that the CPCs belonging to any potential coalition $S$ can adopt. First, whenever the CPC nodes decide to form a coalition $S \subseteq \mathcal{N}$, the CPCs in $S$ would be able to share their kernel estimates of the PUs distributions generated non-cooperatively based on (\ref{eq:discnc}). Hence, within any potential coalition $S$, each CPC can obtain the PUs distribution estimates from its partners and, if deemed suitable, use these distributions as \emph{prior distributions} so as to generate new estimates. Inherently, for a given coalition $S \subseteq\mathcal{N}$, each CPC $i \in S$ must be able to perform the following three steps \emph{for every PU} $k \in \mathcal{K}$:
\begin{enumerate}
\item \textbf{Step~1 - Check Priors Validity:} The first step for each CPC $i \in S$ is to determine, for every PU $k$, whether the prior distributions  received from its cooperative partners in $S \setminus \{i\}$ come from the same distribution as CPC $i$'s own estimate (based on its own  set of observations $\mathcal{L}_{ik}$ for PU $k$).
\item \textbf{Step 2 - Generate New Estimate:} Once a CPC $i \in S$ generates, for any PU $k$, a list of received priors that come from the same distribution as its own estimate (i.e., from CPCs in $S$ that perceive PU $k$'s activity analogously to CPC $i$), its next step is to generate a procedure for combining these received prior distributions with its non-cooperative kernel estimate.
\item \textbf{Step 3 - Assess the Accuracy of the New Distribution:} Given the new estimates generated in Step~2, the last step for CPC $i \in S$ is to assess the accuracy of the resulting distributions.
\end{enumerate}
We will approach the first step using concepts from statistics known as goodness of fit tests while the second step will be performed using a Bayesian nonparametric inference method based on the Dirichlet process. Then, the third step is approached using the concept of a Kullback-–Leibler~(KL) distance. Finally, all three steps will be combined in a single utility function which completes the coalitional game formulation.

\subsection{Priors Validity Check}
Given a coalition $S \subseteq \mathcal{N}$, any CPC $i \in S$ can use \emph{goodness of fit} techniques~\cite{KS00,KS01} to assess whether the set of  kernel estimates received from the CPCs in $S\setminus \{i\}$ regarding the distributions of any PU $k$ come from the same distribution as CPC $i$'s own set of observations $\mathcal{L}_{ik}$.
The goodness of fit of any statistical model provides a description of how well a certain model fits a set of observations or samples~\cite{KS00,KS01}. Goodness of fit tests are one of the most common methods that can be used for identifying whether two sets of observations or samples come from the same distribution or not\footnote{Goodness of fit tests can also be used for other purposes such as comparing an empirical and a theoretical model (see \cite{KS00} for more details.)}.

For the proposed CPC cooperation model, given a CPC $i$ member of a coalition $S$ that receives, from another CPC $j \in S$, a certain kernel estimate $\tilde{H}_{jk}$ for the distribution of PU $k$'s activity, CPC $i$ needs to determine whether   $\tilde{H}_{jk}$ and  its own estimate  $\tilde{H}_{ik}$ are estimates of the same distribution. In other words, each CPC $i$ must identify whether a given cooperating partner CPC $j$ is observing a similar distribution regarding the activity of a certain PU $k$. To do so, CPC $i$ first generates two
sets of samples $\mathcal{H}_{ik}$ and $\mathcal{H}_{jk}$ from $\tilde{H}_{ik}$ and $\tilde{H}_{jk}$, respectively. The samples in $\mathcal{H}_{ik}$ can simply be the original observations $\mathcal{L}_{ik}$ of CPC $i$ or newly generated samples using the continuous kernel estimate $\mathcal{H}_{ik}$. Here, sampling refers to the process of obtaining samples from a distribution function which does not require observing the PU behavior and is commonly performed in  wireless networks.

Then, in order to identify whether $\mathcal{H}_{ik}$ and $\mathcal{H}_{jk}$  come from the same distribution, CPC $i$ uses the \emph{two-sample Kolmogorov-–Smirnov goodness of fit test}, defined as follows~\cite{KS00,KS01}:
\begin{definition}
Consider two sets of observations $\mathcal{H}_{ik}$ and $\mathcal{H}_{jk}$ having, respectively, $h_{ik} =|\mathcal{H}_{ik}|$ and $h_{jk} =|\mathcal{H}_{jk}|$ samples. The \emph{Kolmogorov-–Smirnov statistic} is defined as
\begin{equation}\label{eq:ksstat}
D_{h_{ik},h_{jk}}=\sup_{x}|F_{h_{ik}}(x) - F_{h_{jk}}(x)|,
\end{equation}
where $F_{h_{ik}}$ and $F_{h_{jk}}$ represent the empirical cumulative distribution functions of the samples in $\mathcal{H}_{ik}$ and $\mathcal{H}_{jk}$, respectively. Given $D_{h_{ik},h_{jk}}$, the \emph{two-sample Kolmogorov–-Smirnov~(KS) goodness of fit test} decides that the hypothesis: ``The samples in $\mathcal{H}_{ik}$ and $\mathcal{H}_{jk}$ come from \emph{same} distribution'' is \emph{true} with a significance level $\eta$, if
\begin{equation}
\sqrt{\frac{h_{ik}h_{jk}}{h_{ik}+h_{jk}}}D_{h_{ik},h_{jk}} \le M_{\eta},
\end{equation}
with $M_{\eta}$ a critical value that can be set according to well-defined tables~\cite{KS00}.
\end{definition}

Thus, the two-sample KS test determines whether two sets of samples come from the same distribution or not, without the need for any information on what that distribution actually is. A variety of goodness of fit tests exist, each of which has its own characteristics and practical applications. We have adopted the two-sample KS test due mainly to two reasons~\cite{KS00,KS01}: (i)- It is one of the tests that are most sensitive to differences in both the location and the shape of the empirical cumulative distribution functions of any two sets of observations being compared, and (ii)- It provides a good balance between the complexity and the accuracy of the decision~\cite{KS00,KS01}. For the CPC cooperation problem, this test will be used by each CPC $i$, a member of a coalition $S$, in order to determine whether the estimates received from the CPCs in $S \setminus \{i\}$ come from the same distribution as CPC $i$'s own estimate.\footnote{The CPCs can easily generate a number of samples good enough to ensure the accuracy of the KS test.} As a result, a cooperative CPC $i$ can decide whether a received estimate is valid to be used as a \emph{prior} distribution in order to improve its estimate for some PU $k$.

Subsequently, given any coalition $S$ and any CPC $i\in S$, we let $S_{ik}^\textrm{KS} \subseteq \{S \setminus \{i\}\}$ denote the set of CPCs in $S \setminus \{i\}$  whose estimates regarding the distribution of the activity of PU $k$ have been approved as valid priors by CPC $i$, using the two-sample KS test. Note that, if, for a PU $k$, CPC $i$ could not find any valid prior in $S$, then $S_{ik}^\textrm{KS}=\emptyset$. After the KS test, the next step for any CPC $i \in  S$ is to choose the priors that can potentially improve its estimate of the PUs distributions.

\subsection{A Bayesian Nonparametric Approach for Cooperative Estimate Generation}
 Once a CPC $i$ member of a coalition $S$ determines the set $S_{ik}^\textrm{KS}$ for every PU $k$ using the KS test, this CPC would build a $|S|\times 1$ vector $\tilde{\boldsymbol{H}_k}$ whose elements are the validated priors as received from the CPCs in $S_{ik}^\textrm{KS}$. Given the vector $\tilde{\boldsymbol{H}_k}$, the next step for CPC $i$ is to combine these priors with its own estimate $\tilde{H}_{ik}$ in order to find the posterior distribution, i.e., a new estimate ${H}^S_{ik}(\theta^{L_{ik}+1}_{ik}|\theta_{ik}^{1},\ldots,\theta_{ik}^{L_{ik}})$. To do so, we propose an approach based on Bayesian nonparametric models, namely, using the concept of a \emph{Dirichlet process}~\cite{DP00,DP01,DP02}. The use of such a Bayesian \emph{nonparametric} model, based on Dirichlet processes, is motivated by the following properties~\cite{DP00,DP01,DP02}: (i)- Dirichlet processes are known to be one of the most accurate models that can be applied for modeling \emph{unknown distributions}, (ii)- Dirichlet processes provide \emph{flexible models} that enable one to control the impact of each set of information used in estimation (e.g., the impact of the validated priors), and (iii)- Bayesian nonparametric models can automatically infer an adequate distribution model from a limited data set with little complexity and without requiring an explicit model comparison such as in classical Bayesian approaches.

Before formally defining the Dirichlet process, we must introduce the concept of a \emph{Dirichlet distribution} as follows~\cite{DP00}:

\begin{definition}
Consider a set of events $(X_1,\ldots,X_M)$ that are observed with probabilities $(p_1,\ldots,p_M)$. A \emph{Dirichlet distribution} of order $M \ge 2$ with parameters $(\alpha_1,\ldots,\alpha_M),\ \alpha_i > 0, i \in \{1,\ldots,M\}$ has a probability density function given by:
\begin{equation}
\textrm{Dir}(\alpha_1,\ldots,\alpha_M) = f(p_1,\ldots,p_M;\alpha_1,\ldots,\alpha_M) =  \frac{1}{Z(\alpha_1,\ldots,\alpha_M)}\prod_{i=1}^M p_i^{\alpha_i-1},
\end{equation}
where $\alpha_i-1$ can be interpreted as the number of observations of event $X_i$ and $Z(\alpha_1,\ldots,\alpha_M)$ is a normalization constant given by:
\begin{equation}
Z(\alpha_1,\ldots,\alpha_M)=\frac{\prod_{i=1}^M\Gamma(\alpha_i)}{\Gamma(\sum_{i=1}^M\alpha_i)},
\end{equation}
where $\Gamma(\cdot)$ is the gamma function.
\end{definition}

A Dirichlet distribution $\textrm{Dir}(\alpha_1,\ldots,\alpha_M)$ is the conjugate prior of a multinomial distribution with probabilities $(p_1,\ldots,p_M)$ and can also be seen as a generalization of the beta distribution to the multivariate case~\cite{DP00}. In contrast to a Dirichlet distribution,  a \emph{Dirichlet process} is a stochastic process that is a distribution over probability measures, which are functions that can be interpreted as distributions over a space ${\Theta}$. A draw from a Dirichlet process can be seen as a random probability distribution over the space $\Theta$~\cite{DP00,DP01,DP02}. A distribution over probability measures that is drawn from a Dirichlet process has a marginal distribution that constitutes a Dirichlet distribution. Formally, given a probability distribution $H$ over a continuous space $\Theta$ and a positive real number $\alpha$, the Dirichlet process is defined as follows~\cite{DP00}:

\begin{definition}
A random distribution $G$ on a continuous space $\Theta$ is said to be distributed according to a \emph{Dirichlet process} $\textrm{DP}(\alpha,H)$ with base distribution $H$ and concentration parameter $\alpha$, i.e., $G\sim \textrm{DP}(\alpha,H)$, if
\begin{equation}
(G(\mathcal{A}_1,\ldots,G(\mathcal{A}_r)) \sim\textrm{Dir}(\alpha H(\mathcal{A}_1),\ldots,\alpha H(\mathcal{A}_r)),
\end{equation}
for every finite measurable partition $\{\mathcal{A}_1,\ldots,\mathcal{A}_r\}$ of ${\Theta}$.  The base distribution $H$ represents the mean of the Dirichlet process, i.e., $\operatorname{E}[G(\mathcal{A})]=H(\mathcal{A})$ for any measurable set $\mathcal{A} \subset {\Theta}$ while $\alpha$ is a parameter that highlights the strength of a DP when it is used as a nonparametric prior.
\end{definition}

The Dirichlet process is thus a stochastic process that can be seen as a distribution over distributions, as every draw from a DP represents a random distribution over ${\Theta}$. The base distribution $H$ of a $\textrm{DP}(\alpha,H)$ is often interpreted as a prior distribution over $\Theta$. For the proposed game, the CPCs can use the DP as a means for generating estimates of the PUs distributions by using the validated priors that they received from their cooperating partners. As the real distributions of the PUs are unknown to the CPCs, the CPCs will assume these PUs distributions to be distributed according to a Dirichlet process. Subsequently, each CPC $i$, member of a coalition $S$, needs to combine its own observations about a PU $k$ with the Dirichlet processes received from other cooperating CPCs in the set of validated priors $S_{ik}^\textrm{KS}$.

It is known, from Bayesian nonparametrics, that the combination of a number of independent Dirichlet processes can also be modeled as a Dirichlet process with a strength parameter $\sum_{l\in S_{ik}^\textrm{KS}}\alpha_{lk}$ being the sum of the individual parameters and a prior being the weighted sum of the different priors~\cite{DP00,DP01,DP02}. Thus, from the perspective of any cooperative CPC $i$ the distribution $G_{ik}$ of any PU $k$  is modeled using a Dirichlet process that combines the received estimates from the CPCs in $S_{ik}^\textrm{KS}$ into a single prior, as follows:
\begin{equation}\label{eq:dp}
G_{ik} \sim \textrm{DP}\left(\sum_{l\in S_{ik}^\textrm{KS}}\alpha_{lk},\frac{\sum_{l\in S_{ik}^\textrm{KS}}\alpha_{lk}\tilde{H}_{lk}}{\sum_{l\in S_{ik}^\textrm{KS}}\alpha_{lk}}\right),
\end{equation}
where $\tilde{H}_{lk}$ is the non-cooperative kernel estimate of PU $k$ that CPC $i$ received from a CPC $l \in S_{ik}^\textrm{KS}$ and validated using the two-sample KS goodness of fit test. In (\ref{eq:dp}), the combined strength parameter $\sum_{l\in S_{ik}^\textrm{KS}}\alpha_{lk}$ represents the total confidence level (e.g., trust level in the accuracy of this estimation) in using (\ref{eq:dp}) as a nonparametric prior for inferring on the final distribution. Further, the prior $\frac{\sum_{l\in S_{ik}^\textrm{KS}}\alpha_{lk}\tilde{H}_{lk}}{\sum_{l\in S_{ik}^\textrm{KS}}\alpha_{lk}}$ used in (\ref{eq:dp}) represents a weighted sum (convex combination) of the priors received from the coalition partners. In this combined prior, each weight represents the relative confidence level of a certain prior $\tilde{H}_{lk}$ with respect to the total strength parameter level $\sum_{l\in S_{ik}^\textrm{KS}}\alpha_{lk}$. Note that, the control and setting of the parameters $\alpha_{lk}$ in (\ref{eq:dp}) will be discussed in detail later in this section.

Subsequently, for any coalition $S \subseteq \mathcal{N}$, each CPC $i\in S$ can use the Dirichlet process model in (\ref{eq:dp}) in order to compute a new estimate of the distribution of any PU $k \in \mathcal{K}$ that combines, not only CPC $i$'s own data, but also the data received from CPC $i$'s cooperating partners in $S$. Hence, for every CPC $i$ member of a coalition $S$ having its set of observations $\mathcal{L}_{ik}$ and its vector of validated priors $\tilde{\boldsymbol{H}}_k$, using the Dirichlet process model in (\ref{eq:dp}), the predictive distribution on any new observation $\theta^{L_{ik}+1}_{ik}$ conditioned on $\mathcal{L}_{ik}$ with $G_{ik}$ marginalized out can be given by \cite[Eq.~(5)]{DP01}
\begin{equation} \label{eq:dpest}
{H}^S_{ik}(\theta^{L_{ik}+1}_{ik}\in \mathcal{A}|\theta_{ik}^{1},\ldots,\theta_{ik}^{L_{ik}})=\frac{1}{L_{ik}+\sum_{l\in S_{ik}^\textrm{KS}}\alpha_{lk}}\left(\sum_{l\in S_{ik}^\textrm{KS}}\alpha_{lk}\tilde{H}_{lk}(\mathcal{A}) + \sum_{l=1}^{L_{ik}}\delta_{\theta_{ik}^l}(\mathcal{A}) \right),
\end{equation}
where $\mathcal{A} \subseteq \Theta$. The first term in the parentheses on the right hand side of (\ref{eq:dpest}) represents the contribution of the priors to the distribution ${H}^S_{ik}$ while the second term represents CPC $i$'s kernel estimate based on the observations $\mathcal{L}_{ik}$. Equation (\ref{eq:dpest}) can be re-arranged as follows:
\begin{equation} \label{eq:dpest2}
{H}^S_{ik}(\theta^{L_{ik}+1}_{ik}\in \mathcal{A}|\theta_{ik}^{1},\ldots,\theta_{ik}^{L_{ik}})=\sum_{l\in S_{ik}^\textrm{KS}}w_l\tilde{H}_{lk}(\mathcal{A}) + w_0 \tilde{H}_{ik}(\mathcal{A}),
\end{equation}
where $w_0=\frac{L_{ik}}{\sum_{l\in S_{ik}^\textrm{KS}}\alpha_{lk} + L_{ik}}$ is a weight that quantifies the contribution of CPC $i$'s own kernel estimate in the predictive distribution ${H}^S_{ik}$ and $w_l=\frac{\alpha_{lk}}{\sum_{l\in S_{ik}^\textrm{KS}}\alpha_{lk} + L_{ik}}$ represent weights that identify the strength or impact of the contribution of the priors $\tilde{H}_{lk}$ in the final distribution ${H}^S_{ik}$. The resulting posterior distribution in (\ref{eq:dpest2}) is composed mainly of two terms: a first term related to the received estimates and a second term related to the contribution of CPC $i$'s own observations. The first term in (\ref{eq:dpest2}) reflects the impact of the validated priors received by $i$ (on PU $k$) from the members in $S$ over the final resulting estimate. The second term in (\ref{eq:dpest2}) highlights the contribution of CPC $i$'s own perception of the PU activity. In essence, the weights of both terms in (\ref{eq:dpest2}) are proportional to the length of the observations $L_{ik}$ and to the strength parameters $\alpha_{lk},\ l\in S_{ik}^\textrm{KS}$, of the combined estimate.

One can clearly see from (\ref{eq:dpest2}) that the parameters $\alpha_{lk},\ \forall l \in S_{ik}^\textrm{KS}$, allow the CPCs to control the effect of each prior (as well as the own estimate) on the resulting distribution ${H}^S_{ik}$, depending on the properties of each prior and the confidence that each CPC has in this prior. For example, a CPC $i$ can set the Dirichlet process parameters so as to assign weights of zero to priors that failed the two-sample KS test (i.e., priors outside $S_{ik}^\textrm{KS}$) and, then, it can still use (\ref{eq:dpest2}) for predicting the resulting distribution. Moreover, from (\ref{eq:dpest2}), one can clearly see that by setting all weights $w_l=0,\ \forall l \in S_{ik}^\textrm{KS}$, we obtain the original non-cooperative distribution as estimated by CPC $i$ when acting on its own, i.e., the kernel density estimate of (\ref{eq:discnc}).

In practice, each CPC has an incentive to give a higher weight to priors that were generated out of a larger number of observations, as such priors are more trusted. Hence, in this work, we allow each cooperative CPC $i \in \mathcal{N}$ to set the parameters  $\alpha_{lk},\ \forall l \in S_{ik}^\textrm{KS}$ such that the corresponding weights in (\ref{eq:dpest2}) are proportional to the number of observations, i.e.,
\begin{align}\label{eq:div}
w_0 = \frac{L_{ik}}{\sum_{j \in S}L_{jk}}
\textrm{ and } w_l = \frac{L_{lk}}{\sum_{j \in S}L_{jk}}, \ \forall l \in S_{ik}^\textrm{KS}.
\end{align}
Clearly, by using the definition of the weights as highlighted in (\ref{eq:dpest2}), each CPC $i$ can compute the parameters $\alpha_{lk},\ \forall l \in S_{ik}^\textrm{KS}$ from (\ref{eq:div}).


We note that, in the model studied so far, we assumed that the CPCs have no knowledge of the PU's locations and/or capabilities (beyond a limited number of observations). However, the proposed model can easily extend to the case in which each CPC $i\in\mathcal{N}$ has additional information about the PUs. For example, whenever the CPCs have their own measures of the PUs' locations, they can convert this knowledge into an additional prior that can be combined with the final estimate in (\ref{eq:dpest}) so as to improve the accuracy of the learning process.

\subsection{Utility Function}
Given any coalition $S\subseteq \mathcal{N}$ we define, for every CPC $i \in S$ and for every PU $k\in \mathcal{K}$, the following metric as a measure of the utility yielded from a given estimate of the distribution of PU $k$:
\begin{equation}\label{eq:costk}
u_{ik}(S) = -   \rho\left({H}^S_{ik}(\theta^{L_{ik}+1}_{ik}|\theta_{ik}^{1},\ldots,\theta_{ik}^{L_{ik}}),{H}^S_{ik}(\theta^{(1+\Delta_{ik})L_{ik}+1}_{ik}|\theta_{ik}^{1},\ldots,\theta_{ik}^{(1+\Delta_{ik})L_{ik}})\right),
\end{equation}
where ${H}^S_{ik}$ is given by (\ref{eq:dpest2}), $0 < \Delta_{ik} \le 1$ is a real number, and $\rho(P,Q)$ is the Kullback-–Leibler distance between two probability distributions $P$ and $Q$, given by~\cite{TC00}:
\begin{equation}\label{eq:KL}
\rho(P,Q)= \int_{-\infty}^{\infty}P(x)\log{\frac{P(x)}{Q(x)}}\textrm{d}x,
\end{equation}
where the log is taken as the natural logarithm, hence, yielding a KL distance in \emph{nats}\footnote{Alternatively, one can use a base~$2$ logarithm to get the results in bits.}. The KL distance in (\ref{eq:KL}) is a nonsymmetric measure of the difference between two probability distributions $P$ and $Q$. In a communications environment, the KL distance can be interpreted as the expected number of additional bits needed to code samples drawn from $P$ when using a code based on $Q$ rather than based on $P$~\cite{TC00}. Note that the minus sign is inserted in (\ref{eq:costk}) for convenience in order to turn the problem into a utility maximization problem (rather than cost minimization).

For the proposed game, the utility in (\ref{eq:costk}) measures, using (\ref{eq:KL}), the distance between an estimate of the distribution of PU $k$ when CPC $i$ computes this distribution using $L_{ik}$ observations and an estimate of the distribution of PU $k$ when CPC $i$ uses an extra $\Delta_{ik}L_{ik}$ set of observations to find the estimate. The extra observations $\Delta_{ik}L_{ik}$ can be either observations generated and saved at the beginning of CPC $i$'s operation or newly observed samples. The rationale behind (\ref{eq:costk}) is that, as the estimate of the distribution becomes closer to the real unknown distribution, the KL distance in (\ref{eq:costk}) would decrease, since adding a few more observations to an already exact estimate would yield a little change in this estimate. Hence, as the accuracy of the estimate ${H}^S_{ik}$ improves, the KL distance in (\ref{eq:costk}) would decrease, since the extra $\Delta_{ik}L_{ik}$ observations have a smaller impact on the overall distribution. This method is analogous to iterative techniques used in several statistical domains in which one would stop iterating after observing that a few more iterations have little impact on the final result. As a result, the objective of each CPC $i \in \mathcal{N}$ is to cooperate and join a coalition $S$ so as to maximize (\ref{eq:costk}) by reducing the KL distance $\rho({H}^S_{ik}(\theta^{L_{ik}+1}_{ik}|\cdot),{H}^S_{ik}(\theta^{(1+\Delta_{ik})L_{ik}+1}_{ik}|\cdot))$, on every PU channel $k$. It is interesting to note that (\ref{eq:costk}) allows the CPCs to evaluate the validity of their distribution estimates without requiring \emph{any knowledge} on the actual or  real distribution of the PU.

While cooperation allows the CPCs to improve the estimates of the distributions as per (\ref{eq:dpest2}) and (\ref{eq:costk}), these gains are limited by inherent costs that accompany any cooperative process. These extra costs can be captured by a cost function $c(S)$ which will limit the gains from cooperation obtained in (\ref{eq:costk}). Consequently, for every CPC $i$ member of a coalition $S$, we define the following utility or payoff function that captures both the costs and benefits from cooperation:
\begin{equation}\label{eq:util}
\phi_i(S) = \sum_{k\in \mathcal{K}} u_{ik}(S) - c(S),
\end{equation}
where $u_{ik}(S)$ is given by (\ref{eq:costk}) and $\phi(\emptyset)=0$. The first term in (\ref{eq:util}) represents the sum of KL distances over all PU channels $k \in \mathcal{K}$, as given in (\ref{eq:costk}), while the second term represents the cost for cooperation. Although the analysis done in the remainder of this paper can be applied for any type of cost functions, hereinafter, we consider a cost function that varies linearly with the coalition size, i.e.,
\begin{equation}\label{eq:cfunc}
c(S) = \begin{cases} \kappa  \cdot (|S|-1), & \mbox{if } |S| > 1,\\
0, & \mbox{otherwise,}\end{cases}
\end{equation}
with $0 < \kappa \le1$ representing a pricing factor.  The motivation behind the function in (\ref{eq:cfunc}) is that, in order to perform joint estimation of the PUs' activity, the CPCs that are members of a single coalition $S$ must be able to synchronize their communication and maintain an open channel among themselves to exchange information, share their different observations over time, update the priors of one another, and so on. This synchronization and coordination cost is, indeed, an increasing function of the coalition size such as in (\ref{eq:cfunc}).  The nature of this cost is dependent on the implementation and technology of the CPC network (e.g., whether it is wired, wireless, or heterogeneous). In essence, the linear overhead model would adequately capture most implementations since our approach requires that each CPC $i \in S$ provide a small amount of data. For example, if the nodes are exchanging the data over a wireless channel, each CPC $i \in S$ can broadcast its signalling packet once, to the farthest CPC in $S$. Due to the broadcast nature of the wireless channel, all other members of $S$ would receive this packet. In this case, each CPC $i \in S$ needs simply to wait for the data of the other $|S|-1$ coalition members, and this number would increase linearly with $S$. A similar reasoning can be applied for a wired exchange of data using multicast. Nevertheless, we note that, while the choice of a linear overhead in (\ref{eq:cfunc}) is well-justified, the proposed approach can handle any other models for communication overhead. Note that, when $S$ is a singleton (\ref{eq:util}) would highlight the non-cooperative utility of the CPC in $S$ with no cost. Given (\ref{eq:util}), we can now formally characterize the coalitional game between the CPCs:

\begin{proposition}\label{prop:form}
The proposed CPC cooperation problem is modeled as a coalitional game with non-transferable utility $(\mathcal{N},V)$ in which $\mathcal{N}$ is the set of CPCs and $V(S)$ is a singleton set (hence, closed and convex) that assigns for every coalition $S$ a \emph{single} payoff vector $\boldsymbol{\phi}$ whose elements $\phi_i(S)$ are given by (\ref{eq:util}).
\end{proposition}

Unlike in classical coalitional games, the formation of a single grand coalition encompassing all the CPCs is not guaranteed due to the cooperation costs, as seen in (\ref{eq:util}). In fact, the different CPCs can have their own distinct view of the PUs' activity, and, hence, these CPCs may have no benefit in cooperation. As a result, the proposed CPC coalitional game is classified as a coalition formation game~\cite{WS00} in which the objective is to develop an algorithm that enables the CPCs to cooperate and form coalitions such as in Figure~\ref{fig:ill}.

Note that we are interested in coalition formation games in which the outcome is a set of disjoint CPC coalitions as in Figure~\ref{fig:ill}. The motivation for having disjoint coalitions is two-fold. On the one hand, in a cognitive network, neighboring CPCs having similar views on the PUs' activity would have an incentive to cooperate, and, often, these groups would be disjoint from other groups that have a different PU view (e.g., due to different locations). Moreover, forming disjoint coalitions enables one to exploit significant gains from cooperative estimation as shown later in this paper, while avoiding the significant overhead and complexity associated with having each CPC belong to multiple coalitions. On the other hand, finding low-complexity solutions for coalition formation games with non-disjoint, overlapping coalitions remains an open problem that is, recently, a subject of considerable research in the game theory community~\cite{DR00,Game_theory2,OV01,OV02}. The main reason is that performing coalition formation with multiple membership yields a combinatorial complexity order due to the need for distributing the capabilities of a user among multiple coalitions. In a cognitive network, this translates into a significant degree of complexity for locating and forming coalitions. In a nutshell, forming disjoint coalitions enables one to devise a coalition formation process that optimizes the tradeoff between benefits from cooperation and the accompanying complexity. 


\section{A Distributed Coalition Formation Algorithm}\label{sec:algo}
The CPC coalitional game formulated in Proposition~\ref{prop:form} can be modeled using \emph{hedonic coalition formation games}~\cite{WS00,HC00,HC01} which are a class of coalition formation games in which: (h1)- The players' payoffs depend \emph{only} on the identity of the members in each player's coalition, and (h2)- The formation of coalitions results from a set of \emph{preferences} that the players build over their potential set of coalitions.

By looking at (\ref{eq:util}), we can see that the payoff of any CPC $i \in S$ is dependent solely on the identity of the CPCs of coalition $S$, since the behavior of the CPCs outside $S$, i.e., in $\mathcal{N} \setminus S$, does not impact the utility achieved by the members of $S$ as per (\ref{eq:util}). Thus, the proposed game verifies the first hedonic condition, i.e., condition (h1). In order to cast our game into a hedonic coalition formation game, we need to define preferences for the CPCs, over their possible coalitions.  To do so, it is useful to define the concept of a preference relation or order as follows \cite{HC00}:
\begin{definition}
For any CPC $i\in \mathcal{N}$, a \emph{preference relation} or \emph{order} $\succeq_i$ is a complete, reflexive, and transitive binary relation over the set of all coalitions that CPC $i$ can possibly belong to, i.e., the set $\{S_k \subseteq \mathcal{N} : i \in S_k\}$.
\end{definition}

Hence, for any CPC $i \in \mathcal{N}$, given two coalitions $S_1 \subseteq \mathcal{N}$ and, $S_2 \subseteq \mathcal{N}$ such that $i \in S_1$ and $i \in S_2$, the preference relation $S_1\succeq_i S_2$ implies that CPC $i$ prefers to join coalition $S_1$ rather than coalition $S_2$, or is indifferent between $S_1$ and $S_2$. When using the asymmetric counterpart  $\succ_i$ of $\succeq_i$, $S_1 \succ_i S_2$ implies that CPC $i$ \emph{strictly} prefers joining $S_1$ rather than $S_2$. For the proposed model, the preferences of each CPC must capture this CPC's two, often conflicting, objectives: (i)- Maximize its own individual benefit as quantified by (\ref{eq:util}) and (ii)- Ensure that the overall network benefit, i.e., the social welfare, is maintained at a reasonable level. Inherently, this implies that although the CPCs are mainly interested in optimizing their own utilities as per (\ref{eq:util}), they are also required to ensure that the network's overall performance, i.e., the estimation of the PUs' activity in the whole network, stays at an acceptable level. This second objective is motivated by the fact that, if each CPC acts completely selfishly, the possibility of having inaccurate estimates propagating in the network can increase, which can potentially lead to increased interference to the PUs which can decide, for example, to take specific action against the concerned SUs (e.g., stop them from using the spectrum). This increased interference due to inaccurate estimation may become a detrimental effect that can propagate to all of the cognitive network. Moreover, the CPCs are often owned by the same cognitive network operator, and, thus, they have an incentive not only to improve their own benefit but also the overall network's social welfare. Hence, we propose the following preference relation for any CPC $i\in \mathcal{N}$:
\begin{align}\label{eq:prefsu}
S_1 \succeq_i S_2 \Leftrightarrow q_i(S_1) \ge q_i(S_2) \textrm{ and } v(S_1) + v(S_2\setminus \{i\}) > v(S_1\setminus \{i\}) + v(S_2),
\end{align}
where $S_1,\ S_2 \subseteq \mathcal{N}$, are any two coalitions containing CPC $i$, i.e., $i \in S_1$ and $i \in S_2$, $v(S) = \sum_{j \in S}\phi_j(S)$ is the total utility generated by any coalition $S$, and $q_i:2^{\mathcal{N}}\rightarrow \mathbb{R}$ is a preference function defined as follows:

\begin{align}\label{eq:pref1}
q_i(S) = \begin{cases} \phi_i(S), & \mbox{if } (\phi_j(S) \ge \phi_j(S \setminus \{i\}),\forall j \in S\setminus\{i\}  \\ -\infty, &\mbox{otherwise}, \end{cases}
\end{align}
where $\phi_i(S)$ is the payoff of a CPC $i$ as given by (\ref{eq:util}). The preference function in (\ref{eq:pref1}) implies that the preference value that a CPC assigns to a certain coalition $S$ is equal to the payoff that $i$ achieves in $S$, if the payoffs of the CPCs in $S\setminus \{i\}$ do not decrease when $i$ cooperates with them. Alternatively, the preference value is set to $-\infty$ to convey the fact that if, by being part of a coalition $S$, a CPC $i$ decreases any of the payoffs of the other coalition members in $S\setminus \{i\}$, then,  CPC $i$ will be \emph{rejected} by the members of $S\setminus \{i\}$ and, hence, these players will decide not to form coalition $S$. The usefulness of (\ref{eq:pref1}) in capturing the CPCs objectives will become clearer as we define the following rule that will be subsequently used for developing a coalition formation algorithm:
\begin{definition}\label{def:switch}
\textbf{Join Coalition Rule -} Given a network partition $\Pi=\{S_1,\ldots,S_M\}$ of the CPCs' set $\mathcal{N}$, a CPC $i$ chooses to move from its current coalition $S_m,\ $ for some $m \in \{1,\ldots,M\}$, and \emph{join} a different coalition $S_k \in \Pi \cup \{\emptyset\},\ S_k \neq S_m$, hence forming $\Pi^{\prime} = \{\Pi \setminus \{S_m,S_k\}\} \cup \{S_m\setminus\{i\},S_k\cup\{i\}\}$, if and only if $S_k \cup \{i\} \succ_i S_m,\Pi$. Hence, $\{S_m,S_k\} \rightarrow \{S_m\setminus\{i\},S_k\cup\{i\}\}$ and $\Pi \rightarrow \Pi^{\prime}$. The strict preference $\succ_i$ is formally given by (\ref{eq:prefsu}) with strict inequality.
\end{definition}

The join coalition rule enables every CPC to autonomously decide whether or not to leave its current coalition $S_m$ and join another coalition $S_k \in \Pi$, as long as $S_k \cup \{i\} \succ_i S_m$ as per (\ref{eq:prefsu}). Based on (\ref{eq:prefsu}), a CPC would move to a new coalition if this move can strictly improve its individual payoff and increase the overall utility generated by the two involved coalitions \emph{without} decreasing the payoff of any member of the newly joined coalition (given the \emph{approval} of these other members) as per (\ref{eq:prefsu}) and (\ref{eq:pref1}). 

Subsequently, we develop a coalition formation algorithm consisting of three main phases: PUs monitoring phase, distributed coalition formation phase, and cooperative Bayesian nonparametric estimation phase. In the first phase, before any cooperation occurs, each CPC monitors the activity of all the PUs in its area and records a limited number of observations. Based on these observations, the CPCs use kernel density estimation techniques to generate a non-cooperative estimate on the distributions of the PUs' activity. Once the PUs monitoring phase is complete, the CPCs begin exploring their neighbors in order to identify potential cooperation possibilities. Thus, the distributed coalition formation phase of the algorithm begins. In this phase, the CPCs attempt to identify potential join operations by participating in pairwise negotiations with CPCs (or coalitions of CPCs) in their neighborhood. As soon as a CPC identifies a join operation, based on (\ref{eq:prefsu}), it can decide, in a distributed manner, to switch to the more preferred coalition. We assume that in the coalition formation phase, the CPCs perform their join operations in an arbitrary yet sequential order. This order, in general, depends on the time during which a given CPC requests to perform a join operation. In essence, performing a join operation implies that a CPC leaves its current coalition and coordinates the joining of the new, preferred coalition. The members of the new coalition must give their consent on the joining of every new CPC as captured by (\ref{eq:prefsu}) and (\ref{eq:pref1}). The distributed coalition formation phase is guaranteed to converge to a final partition, as follows:
\begin{theorem}\label{th:one}
 Consider any initial partition $\Pi_{\text{init}}$ that is in place in the cognitive network. The distributed coalition formation phase of the proposed CPC cooperation algorithm will always converge to a final network partition $\Pi_\textrm{final}$ that consists of a number of disjoint coalitions, irrespective of the initial partition $\Pi_{\text{init}}$ .
\end{theorem}
\begin{proof}
Given any initial network partition $\Pi_{\textrm{init}}$, the proposed coalition formation process can be mapped into a sequence of join operations performed by the CPC and which transform the network's partition as follows (as an example):
\begin{align}\label{eq:seq}
\Pi_0=\Pi_{\textrm{init}} \rightarrow \Pi_1 \rightarrow \Pi_2\rightarrow\ldots,
\end{align}
where $\Pi_l=\{S_1,\ldots,S_{M}\}$ is a partition composed of $M$ coalitions that emerges after the occurrence of $l$ join operations. As per (\ref{eq:prefsu}), every join operation performed by a CPC that moves from a coalition $S_1 \in \Pi_{l-1}$ to a coalition $S_2\in \Pi_{l}$, yields:
 \begin{align}\label{eq:ll}
\sum_{j\in S_1}\phi_j(S_1 \setminus \{i\}) + \sum_{j\in S_2}\phi_j(S_2 \cup \{i\}) > \sum_{j\in S_1}\phi_j(S_1) + \sum_{j\in S_2}\phi_j(S_2).
 \end{align}

As the proposed game is hedonic, i.e., the payoff of any CPC $i\in S$ depends only on the identity of the members in $S$, (\ref{eq:ll}), implies that any join operation is accompanied by an increase in the overall social welfare of the network, i.e., $\Pi_{l-1} \rightarrow \Pi_l$ yields
\begin{equation}\label{eq:lll}
\sum_{S \in \Pi_l}v(S) > \sum_{S^\prime \in \Pi_{l-1}}v(S^\prime).
\end{equation}
In this context, (\ref{eq:lll}) implies that each join operation $\Pi_{l-1} \rightarrow \Pi_l$ constitutes a transitive and irreflexive order. Given that the number of partitions of the set $\mathcal{N}$ is finite (given by the Bell number \cite{DR00}), then, the sequence in (\ref{eq:seq}) is guaranteed to converge to a final partition $\Pi_{\textrm{final}}$ which completes the proof.\vspace{-0.1cm}
\end{proof}

Following the convergence of the distributed coalition formation phase, the CPCs begin the last phase of the algorithm which is the cooperative Bayesian nonparametric estimation phase. In this phase, the CPCs monitor, periodically, the PUs' activity while continuously communicating with their cooperative partners and performing the three steps described in Section~\ref{sec:hedonic} for cooperative Bayesian nonparametric estimation of the PUs distributions. In this phase, the CPCs will continue to update their own observations on the PUs' activities while coordinating with their coalition partners so as to constantly improve their estimates of the PUs distributions. Consequently,  whenever the CPC detect that the PUs' activity has changed drastically ( e.g., due to mobility of the PUs),  the involved CPCs can assess whether to reengage in the distributed coalition formation phase in order to adapt the network partition to these environmental changes. A summary of the proposed algorithm is given in Table~\ref{tab:coalform}.

\begin{table}[!t]
  \centering
  \caption{
    \vspace*{-0.7em}The proposed CPCs coalition formation algorithm}\vspace*{-1em}
    \begin{tabular}{p{13cm}}
      \hline
      \textbf{Starting Network} \vspace*{.5em} \\
      \hspace*{1em} The network is governed by a partition $\Pi_{\textrm{init}}=\{S_1,\ldots,S_M\}$ (initially $\Pi_{\textrm{init}}$ = $\mathcal{N}$ = $\{1,\ldots,N\}$ with non-cooperative CPCs). \vspace*{.1em}\\
\textbf{The algorithm consists of three phases}\vspace*{.5em}\\
\hspace*{1em}\emph{Phase~1 - PUs monitoring phase:}   \vspace*{.1em}\\
\hspace*{2em}a) Each individual CPC discovers its neighboring PUs,\vspace*{.1em}\\
\hspace*{2em}b) Each CPC records, over a period of time, a number of observations $L_{ik}$ regarding the
\hspace*{2em}activity of each PU $k$.\vspace*{.1em}\\
\hspace*{1em}\emph{Phase~2 - Distributed Coalition Formation:}   \vspace*{.1em}\\
\hspace*{3em}\textbf{repeat}\vspace*{.1em}\\
\hspace*{4em}a) Each CPC $i \in \mathcal{N}$ performs pairwise negotiations with its surrounding CPCs,\vspace*{.1em}\\
\hspace*{4em}to identify potential join operations.\vspace*{.1em}\\
\hspace*{4em}b) Each CPC $i\in \mathcal{N}$ can identify a join operation by estimating the utility in (\ref{eq:util})\vspace*{.1em}\\
\hspace*{4em}that results, for every PU $k\in \mathcal{K}$, from joining any potential coalition using the\vspace*{.1em}\\
\hspace*{4em}three steps of Section~\ref{sec:hedonic}:\vspace*{.1em}\\
\hspace*{6em}b.1) CPC $i \in S$ checks the validity of the priors on each PU $k$ using the two-sample\vspace*{.1em}\\
\hspace*{6em}Kolmogorov-–Smirnov test.\vspace*{.1em}\\
 \hspace*{6em}b.2) CPC $i$ combines the priors that verify the two-sample Kolmogorov-–Smirnov\vspace*{.1em}\\
\hspace*{6em}test using a Dirichlet process as in (\ref{eq:dpest2}).\vspace*{.1em}\\
 \hspace*{6em}b.3) CPC $i$ computes its potential utility using the Kullback-–Leibler distance in (\ref{eq:util}).\vspace*{.1em}\\
\hspace*{4em}Once the potential utility is found, CPC $i$ identifies whether a join is possible using (\ref{eq:prefsu}).\vspace*{.1em}\\
\hspace*{4em}If a join coalition operation is possible:\vspace*{.1em}\\
\hspace*{6em}a) CPC $i$ leaves its current coalition $S_m$.\vspace*{.1em}\\
 \hspace*{6em}b) CPC $i$ joins a new coalition $S_k$ with the consent of the members of $S_k$ as guaranteed\vspace*{.1em}\\
 \hspace*{6em}by (\ref{eq:prefsu}).\vspace*{.1em}\\
\hspace*{3em}\textbf{until} guaranteed convergence to a Nash-stable partition $\Pi_{\textrm{final}}$.\vspace*{.1em}\\
\hspace*{1em}\emph{Phase 3 - Cooperative Bayesian Nonparametric Estimation:}   \vspace*{.1em}\\
\hspace{2em} This phase occurs inside every formed coalition $S_m \in \Pi_{\textrm{final}}.$\vspace*{.1em}\\
\hspace*{4em}a) Each CPC continues to periodically monitor the PUs' activity.\vspace*{.1em}\\
\hspace*{4em}b) The CPCs part of a same coalition constantly share their updated distribution estimates.\vspace*{.1em}\\
\hspace*{4em}c) The CPCs in every $S_m$ perform cooperative Bayesian nonparametric estimation of the\vspace*{.1em}\\
\hspace*{4em}PUs distributions using the methods of Section~\ref{sec:hedonic}.\vspace*{.1em}\\
\textbf{Based on the results of Phase~3, the CPCs can decide to repeat, periodically, the coalition formation process to adapt to environmental changes such as changes in PU behavior.} \vspace*{.5em}\\
   \hline
    \end{tabular}\label{tab:coalform}\vspace{-0.5cm}
\end{table}

We can study the stability of any network partition $\Pi_{\textrm{final}}$ resulting from the  distributed coalition formation phase of the algorithm in Table~\ref{tab:coalform}, using the concept of a Nash-stable partition defined as follows~\cite{HC00}:
\begin{definition}
A partition $\Pi = \{S_1,\ldots,S_M\}$ of a set $\mathcal{N}$ is said to be \emph{Nash-stable} if $\forall i \in \mathcal{N} \textrm{ s. t. } i\in S_m, S_m \in \Pi,\  (S_m,\Pi) \succeq_i (S_k \cup \{i\},\Pi^{\prime})$ for all $S_k \in \Pi \cup \{\emptyset\}$ with $\Pi^{\prime} = (\Pi \setminus \{S_m,S_k\} \cup \{S_m\setminus\{i\},S_k\cup\{i\}\})$.
\end{definition}

In other words, a partition $\Pi$ is Nash-stable, if no CPC prefers to leave its current coalition and join another coalition in $\Pi$. For the proposed CPC coalitional game, we have the following result:

\begin{proposition}\label{prop:one}
Any partition $\Pi_f$ resulting from the proposed algorithm in Table~\ref{tab:coalform} is Nash-stable.
\end{proposition}
\begin{proof}
First,  as shown in Theorem~\ref{th:one}, the distributed coalition formation phase of the proposed algorithm is guaranteed to converge to a partition $\Pi_{\textrm{final}}$. Assume that this partition $\Pi_{\textrm{final}}$ is not Nash-stable, then, there exists a CPC $i \in S_1,\ S_1 \in \Pi_{\textrm{final}}$ and a coalition $S_2\in \Pi_{\textrm{final}}$, such that $S_2 \cup \{i\} \succ_i S_1 \setminus \{i\}$, i.e., a join operation is possible. Such a case contradicts with the the result of Theorem~\ref{th:one} which ensures that no join operations are possible in $\Pi_{\textrm{final}}$. Thus,  $\Pi_{\textrm{final}}$ must be Nash-stable.
\end{proof}

The proposed algorithm in Table~\ref{tab:coalform} can be implemented by the CPCs of a practical cognitive network, in a distributed manner. Each CPC can discover its neighbors and the surrounding PUs (e.g., using the backbone when the CPC nodes are stations, or, otherwise, by using well-known methods such as those in \cite{ND00}). Then, each CPC can start negotiating with its neighboring CPCs, in a pairwise manner, either over the backbone or over a temporary wireless ad hoc control channel, in order to identify potential join operations. Each CPC can, on its own, determine whether a join operation is possible using (\ref{eq:prefsu}), by performing the three steps for cooperation detailed in Section~\ref{sec:hedonic}. These steps are based on standard mathematical methods that require reasonable computation time. Moreover, implementing these steps does not require any knowledge about the real PU distributions. Once a join operation is identified, the CPC would signal to the new coalition its intention to join. The new coalition is guaranteed to accept this request as (\ref{eq:prefsu}) ensures this approval. Hence, the CPCs would interact, performing distributed join decisions, until reaching a Nash-stable partition. Finally, for any partition $\Pi=\{S_1,\ldots,S_M\}$ of $\mathcal{N}$ that is in place in the network, the computational complexity of finding a potential partner, i.e., identifying a join operation, can be easily seen to be $O(|\Pi|)$ and its worst case scenario is when all the CPCs are acting alone in $\Pi$ in which case $|\Pi| = N$.

\section{Simulation Results and Analysis}\label{sec:sim}
For our simulations, we consider a square area of $2$~km $\times$ $2$~km in which the CPC nodes and the PUs are randomly deployed. The model proposed in this paper does not make any assumptions on the distributions of the PUs' activity, and, thus, it can be applied to any such distributions. In the simulations, we will use \emph{beta distributions}~\cite{DP00} with parameters $\beta_{ik}$ and $\gamma_{ik}$  to describe the activity of a PU $k$ as seen by a CPC $i$. These distributions are generated in such a way that each CPC perceives a different distribution depending on its location with respect to the PU. For the considered beta distributions, $\beta_{ik}-1$ and $\gamma_{ik}-1$ can be interpreted as the number of times PU $k$ is observed to be active and inactive, respectively~\cite{DP00}. To generate these distributions, we need to determine values for $\beta_{ik}$ and $\gamma_{ik}$ such as each CPC $i$ would see a different beta distribution, depending on its location with respect to PU $k$. To do so, first, we assume that, in an ideal case, over a period of $B$~consecutive discrete time instants (e.g., time slots), each PU $k$ transmits for a period of $\tau_k B$ such that $0<\tau_k \le 1$ and is idle for the remaining period. Therefore, in the ideal case, the parameters of PU $k$'s beta distribution can be set to $\tau_kB+1$  and $(1-\tau_k)B+1$.

However, even when the PU is transmitting, depending on the path loss and fading, some CPCs might still see this PU as \emph{inactive}, e.g., if it is already far away and, hence, is not likely to interfere with the SUs that these CPCs are serving. Hence, we consider that a PU $k\in \mathcal{K}$ appears to be \emph{active} at a CPC node $i\in \mathcal{N}$  if its received SNR $\nu_{ik}$ at CPC $i$ is above a certain target threshold $\nu_0$. Otherwise, PU $k$ is considered as \emph{inactive} by CPC $i$, even if it is, in fact, transmitting.
Thus, in order to find how the distributions of the PUs' activity appear to each CPC in the network, the parameters of the  beta distributions must be computed as a function of the received SNR $\nu_{ik}$ of the PU signal at any CPC $i \in \mathcal{N}$.

In this respect, assuming Rayleigh faded wireless channels, we denote by $\chi_{ik}=\operatorname{exp}(-\frac{\nu_0}{\nu_{ik}})$ the probability that the average SNR $\nu_{ik}$ as received by CPC $i$ when PU $k$ transmits is larger than a target value $\nu_0$. The average received SNR from PU $k$ to CPC $i$ is given by $\nu_{ik}=\frac{P_kg_{ik}}{\sigma^2}$ where $P_k$ is the transmit power of the PU,  $\sigma^2$ is the variance of the Gaussian noise, and $g_{ik}=\frac{1}{d_{ik}^\mu}$ is the path loss with $\mu$ the path loss exponent and $d_{ik}$ the distance between PU $k$ and CPC $i$. Consequently, the effective number of times that a PU $k$ is seen to be \emph{active} by a CPC $i$ would be given by $\chi_{ik}\tau_k B$. From the perspective of a CPC $i$, in addition to being effectively inactive for a period of $(1-\tau_k)B$, a PU $k$ is also considered to be \emph{inactive} whenever its SNR drops below $\nu_0$, i.e., with a probability of $1-\chi_{ik}$. Hence, to generate different, yet correlated beta distributions that describe the PUs' activity as perceived by each CPC $i\in \mathcal{N}$, we set the parameters of the distributions to  $\beta_{ik}=\chi_{ik}\tau_k B+1$ and $\gamma_{ik}= (1-\tau_k)B+ (1-\chi_{ik})\tau_k B+1$.

It is important to note that: (i)- The above procedure is assumed to be completely \emph{unknown} to the deployed CPCs, i.e., the CPCs have no knowledge on how they view the distributions of the PUs, and (ii)- The results in this section can also be reproduced for any other types of PUs' activity distributions as well as for other methods for generating these distributions as the proposed model is distribution-independent.

The parameters of the simulations are consequently set as follows. The number of PUs is set to $K=4$ and the pricing factor $\kappa=10^{-3}$, unless stated otherwise. We let $\Delta_{ik}=0.5,\ \forall i\in \mathcal{N}, k \in \mathcal{K}$. The number of observations $L_{ik}$ for a CPC $i$ is assumed to be uniformly distributed over the integers in the interval $[5,20]$. The KS significance level is set to a typical value of $\eta=0.05$~\cite{KS00}.  For the generated distributions, we set $B=10$ and let $\tau_k$ be randomly chosen by each PU $k$ using a uniform distribution over  $[0.4,0.9]$. The transmit power of any PU $k \in \mathcal{K}$ is assumed to be $P_k=100$~mW while the path loss exponent, the Gaussian noise, and the target SNR are, respectively, set to $\mu=3$, $\sigma^2=-90$~dBm, and $\nu_0=10$~dB. All statistical results are averaged over the random locations of the CPCs and the PUs.
\begin{figure}[!t]
\begin{center}
\includegraphics[width=10cm]{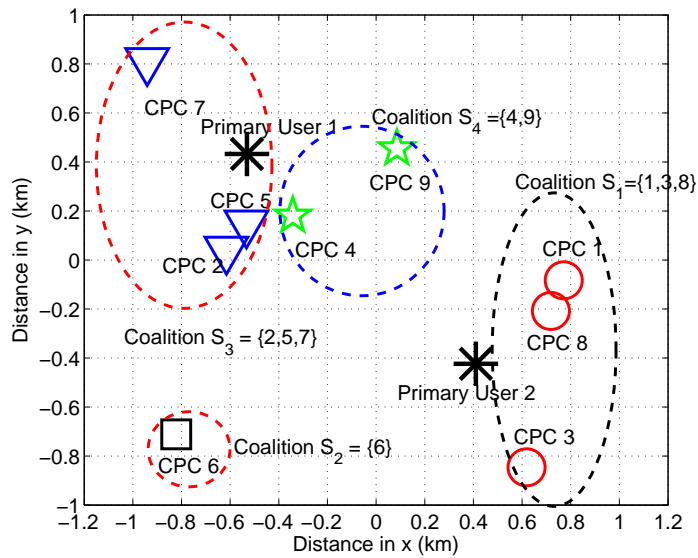}
\end{center}\vspace{-0.8cm}
\caption {A snapshot showing a network partition $\Pi=\{S_1,S_2,S_3,S_4\}$ consisting of $4$ coalitions and resulting from the proposed coalition formation algorithm for a network with $N=9$~CPC nodes and $K=2$~primary users.} \label{fig:snapshot}\vspace{-0.1cm}
\end{figure}

Further, while a body of work (e.g., \cite{Lund,CS00,CS02,CS03,CS04,CS05,CS06,CS07,CS08}) deals with cooperative spectrum sensing techniques, most of this existing work assumes a certain given PU activity distribution. In contrast, in this paper we provide a scheme for learning and estimating the activity distribution of the PUs, from the perspective of a number of cognitive users. In fact, the work done in this paper \emph{complements} cooperative sensing as the distribution of the PUs' activity patterns can serve as an important factor to improve the prediction of cooperative sensing techniques. Thus, given the fundamental difference between the problem solved in this paper and the abundant works on cooperative sensing such as in \cite{CS00,CS02,CS03,CS04,CS05,CS06,CS07,CS08}, a direct comparison of the results is not possible. For instance, to our best knowledge, no existing work has addressed the problem of learning and estimating the statistical distribution of the PUs' activity using cooperative approaches. Thus, we use the commonly used non-cooperative kernel estimation technique as a comparison benchmark. Further, we also compare our results with the real, yet unknown distribution of the PUs' activity so as to provide an additional benchmark for evaluating how close our solution is with respect to an optimal perfect estimate.

\begin{figure}[!t]
\begin{center}
\includegraphics[width=10cm]{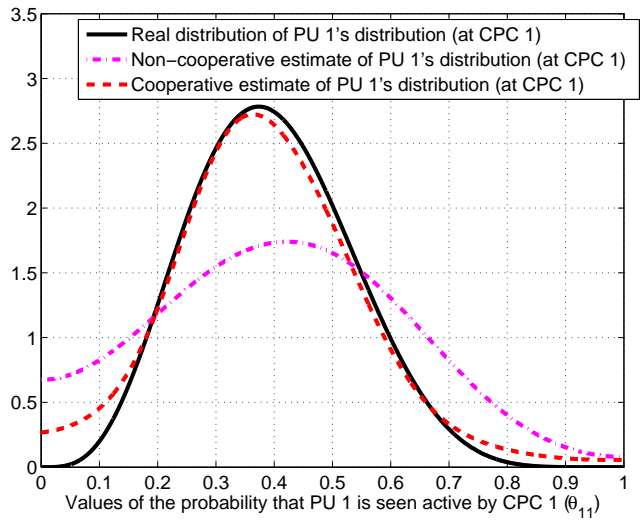}
\end{center}\vspace{-0.8cm}
\caption {The actual real distribution of PU $1$ as seen by CPC node $1$ in the network of Figure~\ref{fig:snapshot} compared to the distributions resulting from the proposed cooperative Bayesian nonparametric approach and from the non-cooperative kernel estimate.} \label{fig:sim}\vspace{-0.1cm}
\end{figure}
In Figure~\ref{fig:snapshot}, we present a snapshot of a partition $\Pi=\{S_1,S_2,S_3,S_4\}$ resulting from the proposed coalition formation game for a randomly generated network having $N=9$~CPC nodes and $K=2$~PUs. Figure~\ref{fig:snapshot} demonstrates how the nodes that are experiencing somewhat similar PUs' activity can decide to form a coalition. For example, consider coalition $S_1$ that consists of CPC nodes $1,3,$ and $8$. In  this coalition, the distribution of PU $1$ is seen by CPC nodes $1,3,8$ as beta distributions with parameters $(\beta_{11}=4.71,\gamma_{11}=7.29),(\beta_{31}=3.72,\gamma_{31}=8.28)$, and $(\beta_{81}=4.6,\gamma_{81}=7.4)$, respectively while the distribution of PU $2$ is seen by CPC nodes $1,3,8$ as beta distributions with parameters $(\beta_{12}=8.21,\gamma_{12}=3.79),(\beta_{32}=8.23,\gamma_{32}=3.77)$, and $(\beta_{82}=8.27,\gamma_{82}=3.73)$. Clearly, CPC nodes $1,3,$ and $8$ have a benefit to use cooperative Bayesian nonparametric estimation to improve their estimate of the distribution of PU $2$ which is seen by all three CPCs with an almost similar distribution (i.e., it passes the KS test for all three CPCs).  However, for PU $1$, although CPCs $1$ and $8$ see a comparable distribution, CPC $3$ has a different view on PU $1$'s activity. In fact, the KS test fails when CPC node $3$ uses it to compare its samples of PU $1$'s distribution to samples from CPCs $1$ or $8$. Nonetheless, all three CPCs find it beneficial to join forces and form a single coalition $S_1$ as it significantly improve their KL distance as per (\ref{eq:util}), on both PUs for CPCs $1$ and $8$, and only on PU $2$ for CPC $3$. Inside $S_1$, CPC $3$ discards the priors received from $1$ and $8$ regarding PU $1$'s distribution and only utilizes the
received priors related to PU $2$ in order to compute its Dirichlet process estimate as in (\ref{eq:dpest2}) for PU $2$. Note that, the partition $\Pi$ in Figure~\ref{fig:snapshot} is clearly Nash-stable as no CPC can improve its utility by switching from its current coalition to another coalition within $\Pi$.

For the network of Figure~\ref{fig:snapshot}, we show, in Figure~\ref{fig:sim}, a plot of the real distribution of PU $1$ as seen by CPC node $1$, compared with the estimates generated from the proposed cooperative Bayesian nonparametric approach and with the non-cooperative kernel estimate. Figure~\ref{fig:sim} clearly shows that, by performing cooperative Bayesian nonparametric estimation, CPC $1$ was able to significantly improve its non-cooperative kernel estimate of PU $1$'s distribution by operating within coalition $S_1$. We note that, in Figure~\ref{fig:snapshot}, the number of non-cooperative observations that CPCs $1,3$, and $8$ record regarding the distribution of PU $1$ are $L_{11}=10,L_{31}=8$ and $L_{81}=20$ observations.  Therefore, Figure~\ref{fig:sim} demonstrates that by using the proposed cooperative Bayesian nonparametric approach while sharing observations (mainly with CPC $8$ in $S_1$), CPC $1$ was able to obtain an almost perfect estimate of PU $1$'s distribution \emph{without any prior knowledge of this distribution} and by using only $L_{11}=10$ own observations. Note that, analogous results can be seen for all CPCs in Figure~\ref{fig:snapshot} as well as for all other simulated networks.

\begin{figure}[!t]
\begin{center}
\includegraphics[width=10cm]{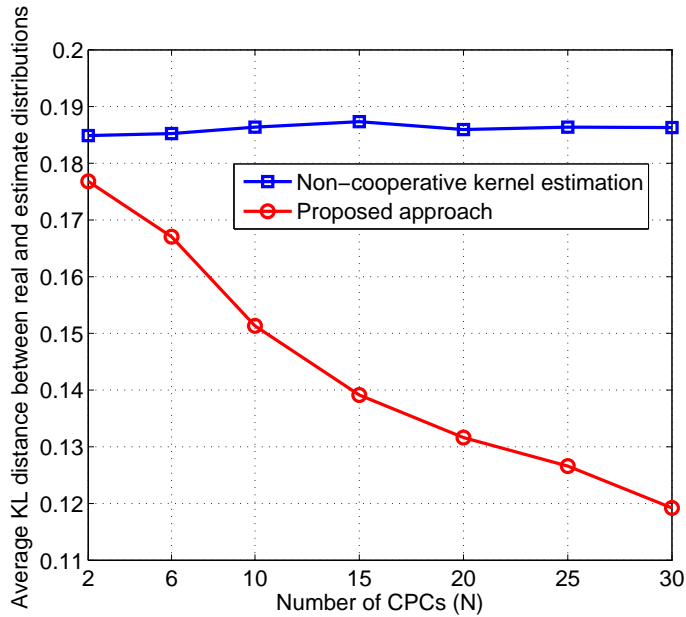}
\end{center}\vspace{-0.8cm}
\caption {Performance assessment showing the average KL distance (per CPC and per PU) between the real distributions (that are unknown to the CPCs) and the estimates generated by the CPCs for a network with $K=4$~PU channels as the number of CPCs $N$ varies.} \label{fig:perf}\vspace{-0.1cm}
\end{figure}
In Figure~\ref{fig:perf}, we assess the performance of the proposed cooperative approach by plotting the average achieved KL distance between the real, yet unknown (by the CPCs), distributions of the PUs and the estimates computed by the CPCs for a network with $K=4$~PUs as the number of CPCs, $N$, varies. This KL distance allows us to assess how accurate the computed estimate is with respect to the actual real PUs' distributions. The results in Figure~\ref{fig:perf} show the average KL distance per CPC and per PU. Figure~\ref{fig:perf} shows that, as the number of CPCs $N$ increases, the  average KL distance between the estimates and the real distributions decreases for the proposed approach and remains comparable for the non-cooperative case. This result demonstrates that, for the proposed approach, as  $N$ increases, the CPCs become more apt to find partners with whom to cooperate and, thus, their performance improves as their estimates become more accurate, i.e., closer to the actual PUs' distributions. Figure~\ref{fig:perf} shows that, at all network sizes, the proposed cooperative approach reduces significantly the KL distance between the real and estimated distributions relative to the non-cooperative case. This performance advantage is increasing with the network size $N$ and reaching up to $36.5\%$  improvement over the non-cooperative kernel estimation scheme at $N=30$~CPCs. Figure~\ref{fig:perf} also shows that our approach allows the average KL distance (average per PU and per CPC) to approach the ideal case of $0$, as more cooperative partners exist in the network, i.e., as the network size $N$ increases.  

\begin{figure}[!t]
\begin{center}
\includegraphics[width=10cm]{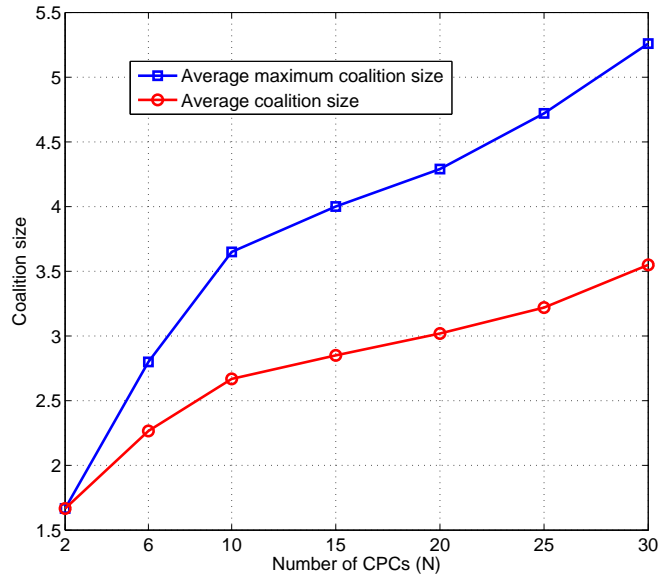}
\end{center}\vspace{-0.8cm}
\caption {Average and average maximum coalition size resulting from the proposed coalition formation algorithm for a network with $K=4$~PU channels as the number of CPCs $N$ varies.} \label{fig:size}\vspace{-0.1cm}
\end{figure}

Figure~\ref{fig:size} shows the average and average maximum coalition size resulting from the proposed algorithm as the number of CPCs, $N$, varies for $K=4$~PUs. In this figure, we can see that both the average and average maximum coalition size are increasing with the network size as cooperation becomes more likely for large networks. From Figure~\ref{fig:size}, we can deduce that the resulting networks are composed of coalitions having small to moderate sizes. In fact,  the average and average maximum coalition size vary from around $1.67$ at $N=2$~CPCs to around $3.5$ and $5.3$, respectively, at $N=30$~CPCs. Hence, Figure~\ref{fig:size} shows that the CPCs self-organize into networks composed of a large number of relatively small coalitions.

\begin{figure}[!t]
\begin{center}
\includegraphics[width=10cm]{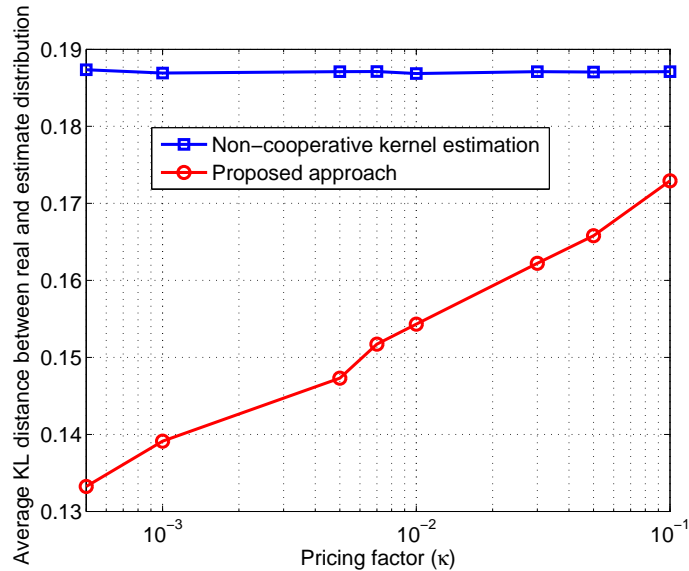}
\end{center}\vspace{-0.8cm}
\caption {Performance assessment showing the average KL distance (per CPC and per PU) between the real distributions and the estimates generated by the CPCs for a network with  $N=15$~CPCs and $K=4$~PU channels as the pricing factor $\kappa$ varies.} \label{fig:price}\vspace{-0.1cm}
\end{figure}

In Figure~\ref{fig:price}, we show the impact of the pricing factor $\kappa$ on the performance of the proposed algorithm in terms of the average KL distance (per CPC and per PU) between the real distributions of the PUs and the estimates computed by the CPCs for a network with $N=15$~CPCs and $K=4$~PUs. Figure~\ref{fig:price} shows that. as the pricing factor $\kappa$ increases, the average KL distance increases as cooperation becomes more costly, hence, limiting the cooperative gains. Nonetheless, Figure~\ref{fig:price} shows that, at all pricing factors, the proposed algorithm maintains a performance advantage relative to the non-cooperative scheme. This advantage, in terms of reduced KL distance with respect to the real PUs' distributions, ranges from around $28.9\%$ at $\kappa=5\cdot10^{-3}$ to about $7.6\%$ at $\kappa=10^{-1}$. The value of $\kappa$ is, in practice, related to the implementation of the network such as the type of backhaul interconnecting the nodes (wired or wireless), the capabilities of the devices, among others. For example, if the PU monitoring is being performed by CPCs connected through a high-speed backhaul, the value of $\kappa$ is expected to be small, e.g., within the range of $10^{-3}$ to $10^{-2}$. In contrast, if the PU activity monitoring is being done by devices with limited capabilities such as cognitive femtocells connected through a DSL backhaul, the cost for information exchange would have a bigger impact and $\kappa$ can have values above $5\%$. We also note that, while the network implementation is the most significant factor in determining $\kappa$, the network operator can use the results of Figure~\ref{fig:price} to have some control over the pricing factor so as to optimize a tradeoff between the potential gains from cooperation and the costs that this cooperation entails, in terms of increased control traffic, communications delay, and overhead. For example, depending on the nature of the CPC nodes' network (e.g., wired or wireless) and their capabilities, the operator can decide to tweak the value of $\kappa$ so as to maintain a certain target QoS requirement during information exchange (e.g., target delay or overhead for signalling) or reserve some backhaul resources for alternate functions.

\begin{figure}[!t]
\begin{center}
\includegraphics[width=10cm]{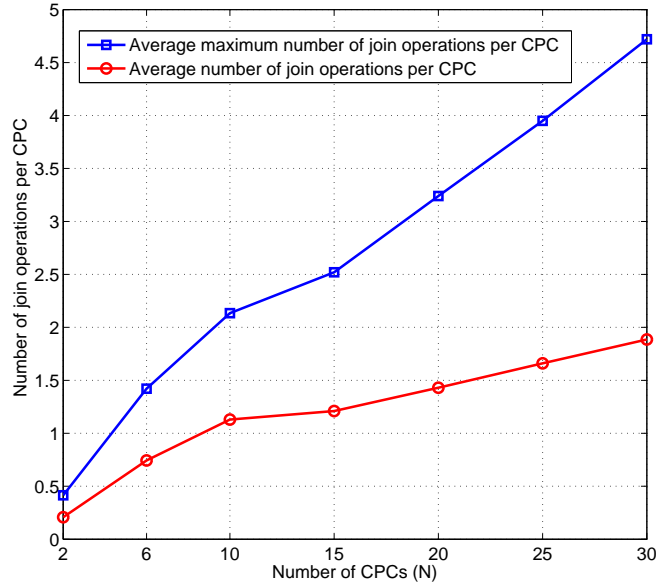}
\end{center}\vspace{-0.8cm}
\caption {Average and average maximum number of join operations attempted per CPC for a network with $K=4$~PU channels as the number of CPCs $N$ varies.} \label{fig:switch}\vspace{-0.1cm}
\end{figure}

In Figure~\ref{fig:switch},  we show the average and average maximum join operations attempted per CPC before convergence of coalition formation as the network size $N$ varies. In Figure~\ref{fig:switch}, we can see that as the number of CPCs $N$ increases, both the average and average maximum number of join operations per CPC increase. This is due to the fact that, as the network size $N$ grows, the possibilities for cooperation increase, and, hence, the coalition formation process yields a larger number of join operations per CPC. Figure~\ref{fig:switch} shows that the average and average maximum number of join operations per CPC vary, respectively, from $0.2$, and $0.4$ at $N=2$~CPCs to $1.9$ and $4.7$ at $N=30$~CPCs. The results in  Figure~\ref{fig:switch} can also be combined with the coalition sizes in Figure~\ref{fig:size} so as to corroborate that the complexity of determining a partner for forming a coalition grows linearly with the size of the network partition in place.

\begin{figure}[!t]
\begin{center}
\includegraphics[width=10cm]{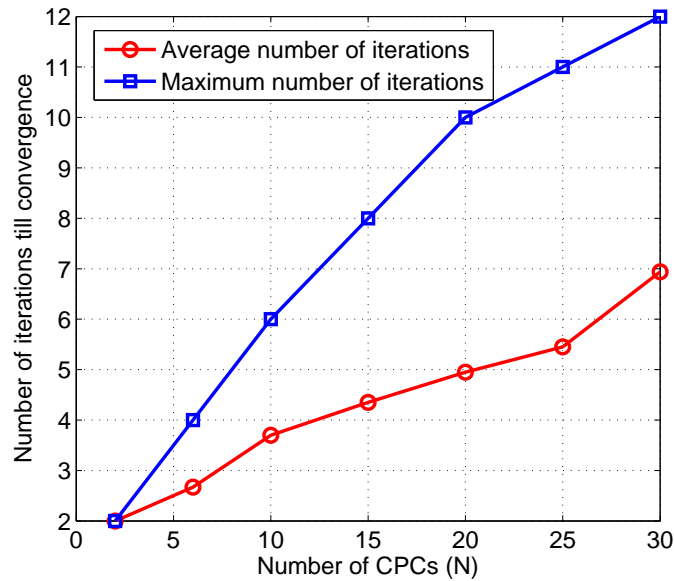}
\end{center}\vspace{-0.8cm}
\caption {Average and maximum number of iterations required till convergence to a Nash-stable partition for a network with $K=4$~PU channels as the number of CPCs $N$ varies.} \label{fig:count}\vspace{-0.1cm}
\end{figure}
The convergence of the algorithm is further assessed in Figure~\ref{fig:count} which shows the average and maximum number of iterations required until convergence to a Nash-stable partition. Each iteration consists of a number of join operations performed by the CPCs. In Figure~\ref{fig:count}, we can see that as the network size $N$ increases, a larger number of iterations is needed for the CPCs to reach a Nash-stable partition. In this respect, the average and maximum number of iterations range from around $2$ at $N=2$~CPCs to $6.94$ and $12$, respectively, at $N=30$~CPCs. Figures~\ref{fig:switch} and \ref{fig:count} clearly show that the proposed algorithm has a low complexity as it enables the CPCs to cooperate, in a distributed manner, while requiring a very reasonable number of iterations and join operations.

To show how the proposed approach can handle changes in the environment, in Figure~\ref{fig:speed}, we plot, as a function of the speed of the PUs, the frequency in terms of average total number of join operations per minute resulting,  over a period of $5$ minutes, from a network with $K=4$~mobile PUs and different number of CPCs $N$. In this figure, the PUs move using a basic random walk mobility model with a constant speed given by the x-axis in Figure~\ref{fig:speed}. Periodically, once the CPCs detect that their view of a certain PU's activity has changed, e.g., due to mobility, they reengage in the coalition formation phase of the algorithm proposed in Table~\ref{tab:coalform}. As a result,  the CPCs may decide to break from their current coalitions and join other coalitions. The increase in the frequency of join operations with the PUs' velocity as seen in Figure~\ref{fig:speed}  is due to the fact that, for more dynamic environments, i.e., higher mobility, the likelihood of the occurrence of join operations increases. Figure~\ref{fig:speed} shows that the average frequency of join operations per minute ranges, respectively for $N=7$~CPCs and $N=15$~CPCs, from $2.2$ and $5.4$~operations per minute at $10$~km/h to about $3.9$ and $8.8$~operations per minute at $100$~km/h. Note that similar results can be seen for other environmental changes such as mobility of CPCs or changes in the PUs transmission pattern $\tau_k$ (for any PU $k$).
\begin{figure}[!t]
\begin{center}
\includegraphics[width=10cm]{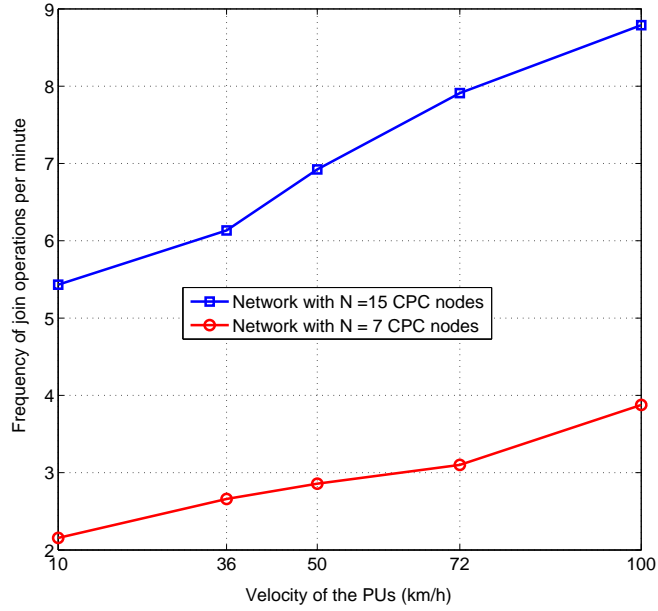}
\end{center}\vspace{-0.8cm}
\caption {Average frequency of join operations per minute as a function of the speed of the PUs achieved over a period of $5$~minutes for a network with $K=4$~mobile PU and different numbers of CPCs.} \label{fig:speed}\vspace{-0.1cm}
\end{figure}

\begin{figure}[!t]
\begin{center}
\includegraphics[width=10cm]{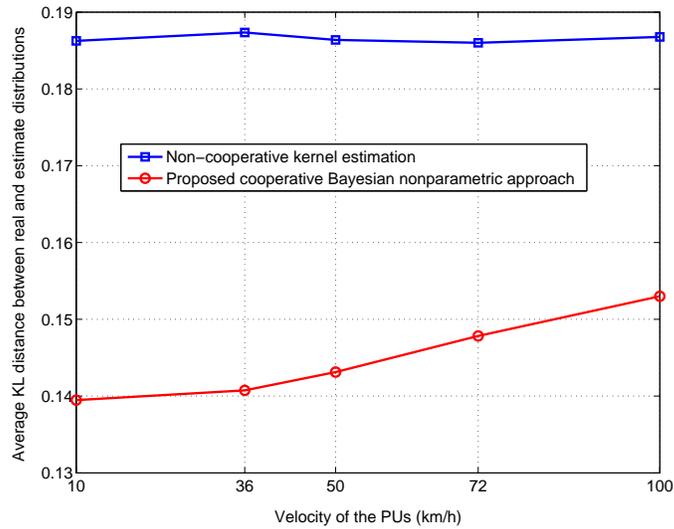}
\end{center}\vspace{-0.8cm}
\caption {Average KL distance between the real distributions and the estimates generated by the CPCs as the speed of the PUs varies over a period of $5$~minutes.} \label{fig:perfspeed}\vspace{-0.1cm}
\end{figure}

In Figure~\ref{fig:perfspeed}, we show how the average KL distance between the real and estimated distributions varies for a network in which the PUs are moving with different speeds. Figure~\ref{fig:perfspeed} shows that, as the speed increases, the average KL distance achieved by the proposed approach increases. This increase is due to the fact that, as the mobility becomes higher, the CPCs become more apt to change their coalitions and, thus, their average KL distance increases due to these changes. Nonetheless, Figure~\ref{fig:perfspeed} shows that the proposed approach maintains its performance advantage, compared to the non-cooperative case, at all PUs' speeds.

\section{Conclusions}\label{sec:conc}
In this paper, we have introduced a novel cooperative approach between the CPC nodes of a cognitive radio network that is suitable for modeling the activity of primary users which is often unknown in practice. Using the proposed cooperative scheme, the CPC nodes can cooperate and form coalitions in order to perform joint Bayesian nonparametric estimation of the distributions of the primary users' activity. We have tackled this problem by formulating a coalitional game between the CPCs and proposing an algorithm for coalition formation. The proposed algorithm allows the CPC nodes to autonomously self-organize into disjoint, independent coalitions. Within each formed coalition, the CPC nodes exchange their non-cooperative distribution estimates and use a combination of Bayesian nonparametric models such as the Dirichlet process and statistical goodness of fit techniques such as the two-sample Kolmogorov-–Smirnov  test, in order to improve the accuracy of the estimated distributions. We have shown the convergence of the proposed algorithm to a Nash-stable partition and we have assessed the properties of the resulting partitions. Simulation results have shown that the proposed algorithm allows a significant improvement in the estimated distribution as quantified by a significant reduction in the Kullback-–Leibler distance between the real, yet unknown (to the CPCs), distributions and the estimates inferred using Bayesian nonparametric techniques. The results also show that the proposed approach enables the CPCs to cope with dynamic changes in their environment. Future work can consider applying the proposed approach for estimating, not only the activity of the primary users, but also the duration of such activity by considering the PUs activity distribution over time. In a nutshell, by marrying concepts from game theory, Bayesian nonparametric estimation, and statistical goodness of fit techniques, we have proposed a novel model for cooperative data estimation that is suitable for many practical applications beyond cognitive networks such as wireless weather services or cooperative multimedia data reconstruction.

\nocite{WS00}
\renewcommand{\baselinestretch}{0.92}
\bibliographystyle{IEEEtran}
\bibliography{references}

\end{document}